\documentclass{article}
\usepackage{ijcai21}

\usepackage{times}
\usepackage{soul}
\usepackage{url}
\usepackage{hyperref}
\usepackage[utf8]{inputenc}
\usepackage[small]{caption}
\usepackage{graphicx}
\usepackage{amsthm}
\usepackage{booktabs}
\usepackage{algorithmic}
\usepackage{amsmath,amssymb,amsfonts}
\usepackage{textcomp}
\usepackage{xcolor}
\usepackage{tikz}
\usetikzlibrary{topaths,calc}
\usepackage[ruled]{algorithm2e}
\usepackage{subfigure}

\newtheorem{myTheo}{Theorem}

\title{A Game Generative Network Framework with Its Application to Relationship Inference}

\author{
Jie Huang$^1$
\and
Fanghua Ye$^2$\and
Xu Chen$^3$
\affiliations
$^1$University of Illinois at Urbana-Champaign\\
$^2$University College London\\
$^3$Sun Yat-sen University\\
\emails
jeffhj@illinois.edu,
smartyfh@outlook.com,
chenxu35@mail.sysu.edu.cn
}

\begin{document}

\maketitle

\begin{abstract}
A game process is a system where the decisions of one agent can influence the decisions of other agents. In the real world, social influences and relationships between agents may influence the decision makings of agents with game behaviors. And in turn, this also gives us the opportunity to mine such information from agents by the observed interactions of them in a game process. In this paper, we propose a Game Generative Network (GGN) framework that utilizes the deviation between the real game outcome and the ideal game model to generate networks for game processes, which opens a door for understanding agents with game behaviors by network mining approaches. We illustrate how to apply GGNs to infer the hidden relationships between agents with game behaviors and conduct experiments on team games as a concrete example. Experimental results demonstrate that our proposed framework can reveal the hidden relationships of agents in such games.
\end{abstract}

\section{Introduction}

In the real world, there are a variety of systems with cooperation and competition behaviors among agents. Such behaviors are usually modeled as games where the decisions of one agent can influence the decisions of other agents \cite{osborne1994course,myerson2013game}, and the decisions of agents may also be affected by the relationships between them \cite{chen2014social}. 
However, due to the complexity of relationships between agents -- difficult to obtain and model, most game-theoretic models do not consider relationships between agents comprehensively or even ignore such information.
If we want to better understand agents with game behaviors, it is necessary for us to mine the social influences and relationships between them. Also, by incorporating such information, we can build game models that better characterize agents' real game behaviors.

On the other hand, the pattern of connections and relationships between objects in a system can be represented as a network \cite{newman2003structure}, where the objects correspond to nodes and the relationships between objects correspond to edges. A game is a system of multiple agents with special relationships, either explicit or implicit. Thus it is possible for us to generate a network for a game process and mine the relationships between agents on the network. Although there are a large number of studies on social network mining, none can be directly applied to social agents with game behaviors. Generating networks for game processes will no doubt open a door for mining hidden information of agents with game behaviors.

In addition, relationships between agents in game processes are usually more complex than common relationships modeled by traditional networks. In addition to positive relationships, there may also be some negative relationships between agents in a game system, which greatly influence the strategies taken by agents. In a game process, positive and negative relationships reflect in many different ways. For instance, a positive relationship may represent friendship, cooperation, trust or win-win and a negative relationship may represent hostility, competition, distrust or loss-loss \cite{guha2004propagation,brzozowski2008friends,pang2008opinion}. Compared to traditional networks that reduce the relationships between objects as simple pairwise links, signed networks \cite{tang2016survey} in which the weights of edges can be positive and negative are better to represent such positive and negative relationships between agents in game processes. 

Based on the above considerations, we propose a general Game Generative Network (GGN) framework, which generates signed networks for game processes based on the deviation between the real game outcome and the ideal game model. Our assumption for such games is that if there are some relationships between agents, the relationships will affect the strategies that agents take.
In general, GGNs serve as the intermediaries between game processes and network mining approaches.
In this paper, we mainly focus on the relationship inference task. 
In spite of the generality of GGNs with various potential applications, 
we first give a simple illustration of the Prisoners' Dilemma \cite{rapoport1965prisoner}, and then introduce a game generative network mining approach and apply it to infer the relationships between social agents in team games.

The main contributions of this paper are summarized as follows:
\begin{itemize}
  \item We appear to be the first to mine positive and negative relationships between social agents with game behaviors from game processes.
  \item We propose a Game Generative Network framework, which opens a door for analyzing hidden relationships of agents with game behaviors via network mining approaches.
  \item We study team game and apply our framework to it.
  Experimental results demonstrate the feasibility of our GGN framework on revealing relationships between agents in such games.
\end{itemize}

The code and data are available at \url{https://github.com/jeffhj/GGN}.

\section{Game Generative Network}

\subsection{Preliminaries}

\subsubsection*{Game Theory} 

Game theory is a tool to analyze decision making in multi-agent interactions. Based on different analyzing methods, there are different types of game models designed for different specific games. 
A \textit{game} denoted as $\Gamma$ usually includes three parts: the set of agents $\mathcal{P}$, 
the strategy space $\mathcal{S}$, and the utility function $u: \mathcal{S} \to \mathbb{R}$. If the number of agents in a game is $n$, we call it $n$-agent or $n$-player game. A \textit{(pure) Nash equilibrium} in an $n$-agent game is a list of strategies $\textbf{s}^* = \{s^*_1, s^*_2, \dots, s^*_n\}$ such that 
\begin{equation}
s_i^* = \mathop{\text{argmax}}_{s_i} u_i(s^*_1,s^*_2,\dots,s_i,\dots,s^*_n),
\end{equation}
where $s_i$ is the strategy taken by agent $p_i$. In other words, Nash equilibrium is a stable strategy list that neither agent can increase her payoff by taking another strategy, thus no agent will try to change her strategy.

\subsubsection*{Signed Network} 

A signed network is defined as a graph $\mathcal{G}(\mathcal{V},\mathcal{E},w)$, where $\mathcal{V}$ is the node set, $\mathcal{E}$ is the edge set, and $w: \mathcal{E} \to \mathbb{R}$ is a weight mapping function associated with each edge and the weight can be either positive or negative. 

\subsection{A Special Case: Dynamic Game}

We first study a special type of game named dynamic game. \textit{Dynamic game} is a kind of game in which decisions of agents are made at various times with some of the earlier decisions being public knowledge when the later decisions are being made. Here we consider the $n$-agent dynamic game with agent set $\mathcal{P} = \{p_1,p_2,\dots,p_n\}$, where each agent has multiple opportunities to change her strategies and the payoff updates after each strategy changes. 

In order to involve the interactions between agents, we build a network with node set $\mathcal{V} = \mathcal{P}$. The original strategy list of each agent at time $t-1$ is $\textbf{s} = \{s_1, s_2, \dots, s_i, \dots, s_n\}$. At time $t$, agent $p_i$ changes her strategy from $s_i$ to $s'_i$, which converts $\textbf{s}$ to {$\textbf{s}' = \{s_1, s_2, \dots, s'_i, \dots, s_n\}$}. 
Since the strategies made by agents in dynamic games are based on some public knowledge, agents are aware of others' strategies and can make decisions according to such knowledge. 
Thus relationships between agents may be reflected dynamically in such interactions.
Based on this observation, we use the utility difference before and after the conversion of each agent to build the edges of the network, the weight of each edge is formulated as follows:
\begin{equation}
w(e^{(t)}_{i,j}) = u_j(\textbf{s}') -  u_j(\textbf{s}),
\end{equation}
\noindent where $u_j(\textbf{s})$ is the utility of agent $p_j$ with strategy list $\textbf{s}$, and $e^{(t)}_{i,j}$ is the directed edge from $p_i$ to $p_j$ built at time $t$. Repeating the above process for $T$ times, we will end up with a directed signed network containing the interaction information of agents.

\subsection{General Game Generative Network Framework}

For the special case in the previous section, it is natural to generate a network implying the relationship information of agents. However, in most cases, we cannot get the whole game process. This is to say, in a real game, we may only get the final outcome, i.e., the real strategies taken by agents. Besides, such dynamic games are not universal. So a more general game generative network framework is necessary.

Most existing game-theoretic models assume that all the agents are in selfish behaviors, where each agent aims at maximizing her own utility. Such assumptions ignore social influences and relationships between agents. Besides, sometimes relationships among agents will cause a deviation in the utility of each agent. On one hand, most game models either do not consider relationships between agents or do not consider such information comprehensively. One the other hand, such information of agents is also difficult to obtain and hard to characterize by game models.

As discussed above, an ideal game model is based on the assumption that each agent is selfish and usually does not consider the complete relationships between agents, but the real outcome does not follow such assumptions and reflects some unobserved relationships. For instance, from an idea game model, we may find that both agent $p_1$ and agent $p_2$ can achieve a higher payoff by reducing the other's payoff, but the real outcome shows that they did not do so (we will use Prisoners' Dilemma as a concrete example in the following section). This may be because there is a positive relationship between them, and such a relationship will be reflected in the deviation between the real game outcome and the ideal game model.
Therefore, in order to generate networks for game processes, one practical approach is to use the deviation between the real game outcome and the ideal game model.

Assume the real strategy list is $\textbf{s}_r = \{s_1^r,s_2^r,\dots,s_i^r,\dots,s_n^r\}$. For agent $p_i$, we can predict the strategy $s^*_i$ taken by $p_i$ from an ideal game model. By replacing $s_i^r$ with $s_i^*$, we get a new strategy list $\textbf{s}' = \{s_1^r,s_2^r,\dots,s_i^*,\dots,s_n^r\}$. And then we use the utility difference before and after the conversion of each agent to build the edges of the network, the weight of each edge is formulated as follows:
\begin{equation}
w(e_{i,j}) = u_j(\textbf{s}_r) - u_j(\textbf{s}').
\label{eq:ggn_weight}
\end{equation}

There are two ways to choose $s_i^*$: one way is to use the strategy in a pure strategy Nash equilibrium. If there are multiple pure strategy Nash equilibria, we can also build multiple networks. But the concern is that the Nash equilibrium may be hard to find and pure strategy Nash equilibrium does not exist in some cases. The second way is to choose the strategy that maximizes the utility of $p_i$ as follows:
\begin{equation}
s_i^* = \mathop{\text{argmax}}_{s^{\#}_i} u_i({s_1^r,s_2^r,\dots,s^{\#}_i,\dots,s_n^r}).
\end{equation}

Based on the above framework, we can generate signed networks for most game processes, namely Game Generative Networks (GGNs).

\section{Relationship Inference on Game Generative Network}

To show how to apply GGNs to solve downstream tasks, we study relationship inference with GGNs.
Relationship inference aims to infer the relationships between agents. And specifically, in this paper, we aim at judging whether the relationship between two agents in a game is positive or negative.
In some games, the relationships between agents may influence agents' decision makings. For instance, some positive relationships (e.g., friendship, cooperation, trust, and win-win) and negative relationships (e.g., hostility, competition, distrust, or loss-loss) may be involved in game processes, and the utilities of agents are affected by such relationships. But most of the time, we have no knowledge about relationships between agents or the signs of such relationships. Thus relationship inference on social agents with game behaviors is quite interesting and useful if we want to further understand agents' behaviors and the relationships between them.

\begin{figure}[tp!]
\centering
\includegraphics[width=\linewidth]{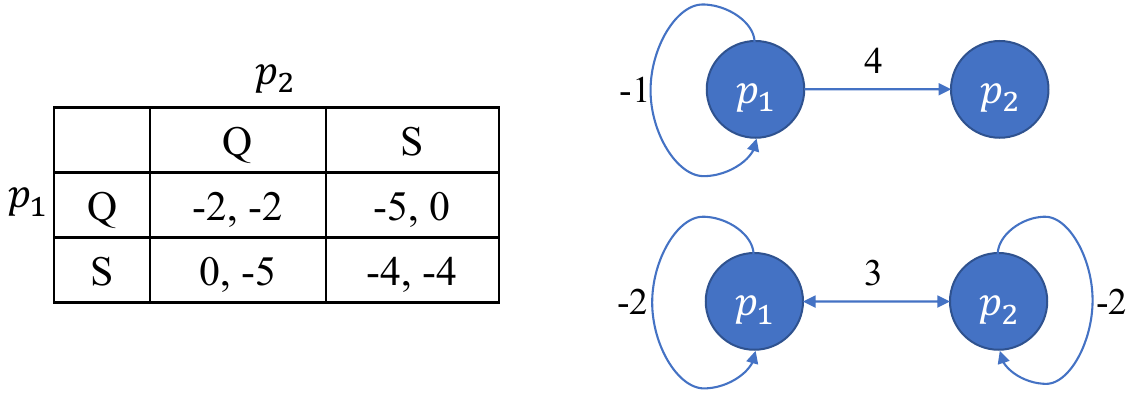}
\caption{The illustration of Prisoners' Dilemma.}
\descriptionlabel{}
\label{fig:prisoners}
\end{figure}

\subsection{Prisoners' Dilemma}

Here we give a simple illustration of relationship inference with GGNs. We study the famous problem of Prisoners' Dilemma \cite{rapoport1965prisoner}. The problem is defined as follows: two criminals ($p_1$ and $p_2$) are imprisoned and they cannot talk to each other. The police offer each prisoner a bargain -- keep quiet (Q) or to squeal (S). The payoff table of these two criminals is shown in Fig. \ref{fig:prisoners} (left). For instance, if $p_1$ chooses strategy S and $p_2$ chooses strategy Q, $p_1$ will be released and $p_2$ will be imprisoned for 5 years. 

In the ideal game model where the relationship between these two agents is ignored, Nash equilibrium of Prisoners' Dilemma is that both criminals choose to squeal thus both of them will be imprisoned for 4 years. But if we consider the relationship between $p_1$ and $p_2$, the real strategies they take may not be consistent with the Nash equilibrium of an ideal model.

For example, suppose that the real outcome of the Prisoners' Dilemma game happens as follows: $p_1$ chooses to keep quiet and $p_2$ chooses to squeal. Based on the framework of GGN, we can generate a network of $p_1$ and $p_2$ in Fig. \ref{fig:prisoners} (upper right). 
For example, $w(e_{1,2}) = u_2(\{Q,S\}) - u_2(\{S,S\}) = 0 - (-4) = 4$.
From the network, we can conclude that $p_1$ has a positive relationship to $p_2$, that is $p_1$ is doing favor towards $p_2$. Another possible situation is that both $p_1$ and $p_2$ keep quiet. From the GGN illustrated in Fig. \ref{fig:prisoners} (lower right), we can conclude that there is a strong relationship between $p_1$ and $p_2$ so that they believe in each other. These examples illustrate that GGNs can be useful to reveal relationships between agents.

\begin{figure*}[tp!]
\centering
\includegraphics[width=\linewidth]{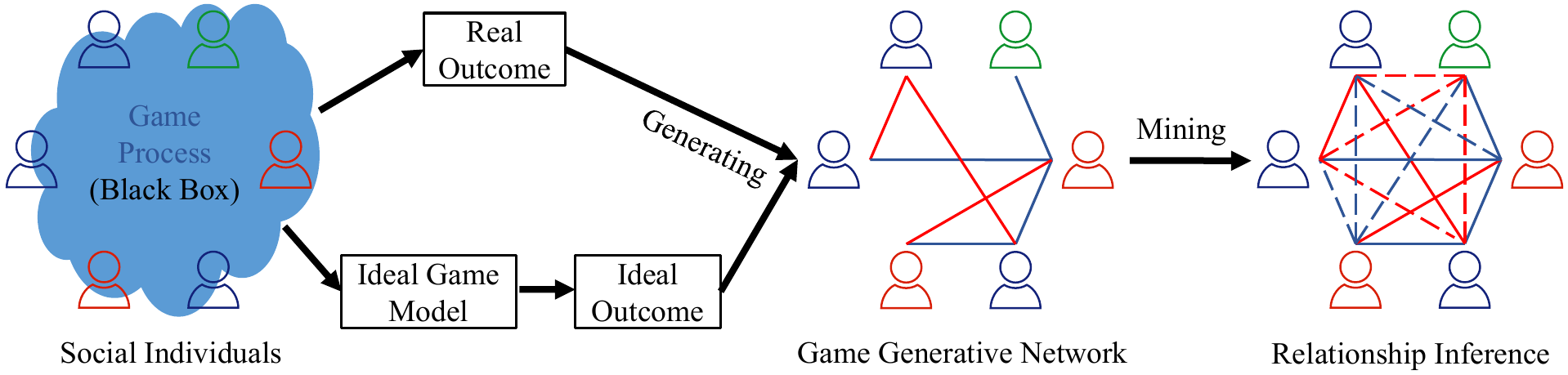}
\caption{The process of relationship inference with GGN, where GGN serves as an intermediary that combines game processes and network mining approaches. The color of positive edges in the GGN is red and the color of negative edges is blue.}
\descriptionlabel{}
\label{fig:ggn_model}
\end{figure*}

\subsection{Game Generative Network Mining}

For small GGNs (e.g., $n \leq 10$), we can conclude the relationships between agents by observation or simple statistical analysis. For big GGNs, it is necessary to apply an automatic network mining method. From our definition above, a GGN is essentially a signed network. With such cognition, the problem of mining for GGNs is indeed the problem of mining for signed networks \cite{leskovec2010predicting,kunegis2009slashdot,tang2016survey}. Since the focus of this paper is on exploring the capabilities of GGN, we design a relatively simple and intuitive method to mine for GGN and observe the power of GGN by using this method.

The most direct relationship between agents in GGNs is the first-order proximity. For each pair linked by an edge $e_{u,v}$ with weight $w_{uv}$, if $w_{uv}>0$, $e_{u,v}$ represents a positive relationship from $u$ to $v$; if $w_{uv}<0$, $e_{u,v}$ represents a negative relationship from $u$ to $v$; otherwise, there is no first-order proximity between $u$ and $v$ represented by $e_{u,v}$. Since not all the edges generated are meaningful, we should filter out some edges based on specific cases. For instance, given a relationship network of agents, we have prior knowledge about whether two agents know each other, thus the edges beyond the edges in the relationship network, including self-loops, should be filtered out. 

Besides, the first-order proximity cannot represent all relationships between agents, and because the behaviors of agents are not fully reflected in a game, some edges may be missed. Based on the multiplicative transitivity \cite{kunegis2009slashdot} of signed networks, if there is an edge with weight $w_{uv}$ between $u$ and $v$, and there is an edge with weight $w_{vx}$ between $v$ and $x$, then there may be a hidden third edge with weight $w_{uv}\cdot w_{vx}$ between $u$ and $x$. Multiplicative transitivity is proved by the fact that two agents connected by an even number of negative edges can be considered balance \cite{hage1983structural}. 

Based on the multiplicative transitivity, we introduce the exponential kernel to evaluate the $k$th-order relationships between agents. The $k$th-order exponential kernel is defined as the weighted sum of matrix powers, where the weight decays with the inverse factorial:
\begin{equation}
\label{exp_k}
\text{exp}_k(A) = \sum_{i=0}^{k}\frac{1}{i!}A^i.
\end{equation}

\noindent The sign of the $(i,j)$-th entry of the exponential kernel indicates the positive or negative relationship from $i$ to $j$. The greater the absolute value, the stronger the relationship. Combining with the generative process of GGN, the process of relationship inference with GGN is illustrated in Fig. \ref{fig:ggn_model}.

\section{Game Generative Network for Team Game}

\label{sec:team_game}

In spite of the generality of GGNs that can apply to a variety of games, for an empirical study on a specific game, we will show how to apply GGNs to the team game for relationship inference as a concrete example. The goal of team game is to divide $n$ agents into small groups. The utility of each agent in team games is affected by the relationships between agents. Relationships between agents in team games have a broad definition. For instance, a positive or negative relationship can represent a friend or a foe relation between two agents. In this case, everyone wants to team up with her best friends and tries not to team up with those she dislikes. On the other hand, relationships in team games can evaluate whether the cooperation between two agents will bring positive (win-win) or negative (loss-loss) effects. And in this case, everyone wants to team up with a group of agents that make her obtain the biggest profits.

\subsection{Team Game Definition}

In this section, we design an ideal game model for the team game. Consider an undirected signed network $\mathcal{G}(\mathcal{V}, \mathcal{E}, w)$ with the node (agent) set $\mathcal{V}$, the edge set $\mathcal{E}$ and the weight $w(e)$ associated with each edge $e$. 
In our settings here, team game aims to find a set of teams and each node has one and only one team, and we use $t(i)$ to denote the team which node $i$ belongs to. The strategy of agent $i$ is to quit the current team $t(i)$ and join a new team $t'(i)$.

From the above analysis, we can conclude that each agent wants to team with those agents who have positive relationships with her, thus we define the gain function as follows:
\begin{equation}
g_i(\textbf{s}) = \sum_{j=1,j\neq i}^n A_{ij}\delta(i,j),
\end{equation}
\noindent where $n$ is the number of nodes, $A_{ij}$ is the $(i,j)$-th entry of the weighted adjacent matrix $A$ of the signed network $\mathcal{G}(\mathcal{V}, \mathcal{E}, w)$ above, and $\delta(i,j)$ is an indicator function which is equal to $1$ when $t(i) = t(j)$ and $0$ otherwise.

However, when most of the relationships of agents are positive, such gain function will lead to a trivial solution where all the agents form a single team. To solve this problem, we consider that the number of members in each group should be as balanced as possible, which is consistent with the fact that a team with a large size will weaken the relationships. To this end, we define the loss function as $l_i(\textbf{s}) = c|t(i)|$,
\noindent where $|t(i)|$ is the size of the team which agent $i$ belongs to and $c$ is a parameter to balance the gain and the loss. Finally, the utility of agent $i$ with strategy list $\textbf{s}$ is calculated as follows:
\begin{equation}
u_i(\textbf{s}) = g_i(\textbf{s}) - l_i(\textbf{s}) = \sum_{j=1,j\neq i}^n A_{ij}\delta(i,j) - c|t(i)|.
\end{equation}

\begin{algorithm}[tp!]
    \SetAlgoNoEnd
    \caption{Team Game($\mathcal{G}$)}  

    Initialize each node to a singleton team\;
    \While{not converge}
    {
      Random shuffle the nodes\;
      \For{$i$ \textbf{from} $1$ \textbf{to} $n$}
      {
        $s_i = \text{argmax}_{s_i^*} u_i(\{s_1, s_2, \dots, s_i^*, \dots, s_n\})$\;
        Node $i$ joins the new team by taking strategy $s_i$\;
      }
    }
    \label{LocalEquilibrium}
\end{algorithm}

\paragraph{Existence of Nash Equilibria}

We will prove that the team game we designed is a finite exact potential game which always possesses pure strategy Nash equilibria \cite{monderer1996potential}. 
Firstly, let us recall the definition of the exact potential game.
A game is an \textit{exact potential game} if there exists an associated potential function $\Phi(\cdot)$ defined on the strategy profiles that satisfies $\Phi(s'_i,\textbf{s}_{-i}) - \Phi(s_i,\textbf{s}_{-i}) = u_i(s'_i,\textbf{s}_{-i}) - u_i(s_i,\textbf{s}_{-i})$ for every strategy profile $\textbf{s}_{-i}$ of all agents except agent $i$ and every strategy $s_i$ of agent $i$.  In an exact potential game that contains a finite number of strategy profiles, Nash equilibria always exist. And every better response in which each agent sequentially changes her strategy to improve her own utility, will finally converge to a Nash equilibrium \cite{monderer1996potential}.

\begin{myTheo}
The team game with utility function $u_i(\cdot)$ is an exact potential game with potential function $\Phi(\textbf{s}) = \sum_{i=1}^n \frac{1}{2}(g_i(\textbf{s}) - l_i(\textbf{s}))$.
\end{myTheo}

\begin{proof}
\begin{align}
&\Phi(s'_i,\textbf{s}_{-i}) - \Phi(s_i,\textbf{s}_{-i}) \notag\\
= & \frac{1}{2}(\sum_{j=1,j\neq i}^n A_{ij}\hat{\delta}(i,j) - \sum_{j=1,j\neq i}^n A_{ij}\delta(i,j)) \notag\\
& - \frac{c}{2}((|t'(i)|+1)^2+(|t(i)|-1)^2 - |t'(i)|^2-|t(i)|^2) \notag\\
= & g_i(s'_i,\textbf{s}_{-i}) - g_i(s_i,\textbf{s}_{-i}) - c(|t'(i)|+1) + c|t(i)| \notag\\
= & g_i(s'_i,\textbf{s}_{-i}) - g_i(s_i,\textbf{s}_{-i}) - l_i(s'_i,\textbf{s}_{-i}) + l_i(s_i,\textbf{s}_{-i}) \notag\\
= & u_i(s'_i,\textbf{s}_{-i}) - u_i(s_i,\textbf{s}_{-i}). \notag
\end{align}
\end{proof}

The team game is a game with finite agents and finite strategy space, so the team game with utility function $u_i(\textbf{s})$ is a finite exact potential game, thus possesses pure strategy Nash equilibria.

Based on the finite improvement property of potential game, we propose an algorithm for computing the Nash equilibrium of the ideal team game model, which is shown in Algorithm \ref{LocalEquilibrium}. We firstly initialize each node to a singleton team and then repeat the following process until the game converges: random shuffle the nodes, and then let each node quit the current team and join the team which maximizes her utility in the current state.

\subsection{Experiments}

\begin{table}[tp!]
    \begin{center}
    \begin{tabular}{l|r|r|r|r}
        \toprule
        Datasets & Nodes & Edges & + edges & - edges \\
        \midrule
        Slashdot1 & 3,869 & 93,498 & 77,052 & 16,446  \\ 
        \hline
        Slashdot2 & 3,872 & 27,298 & 20,134 & 7,164 \\
        \hline
        Epinions1 & 6,605 & 182,674 & 158,170 & 24,504 \\
        \hline
        Epinions2 & 6,591 & 36,992 & 32,832 & 4,160 \\
        \bottomrule
    \end{tabular}
    \end{center}
    \caption{Statistics of the datasets.}
    \label{table:dataset}
\end{table}

\begin{figure}[tp!]
\centering
\subfigure[]
{
\begin{minipage}[t]{0.48\linewidth}
\centering
\includegraphics[width=4.3cm]{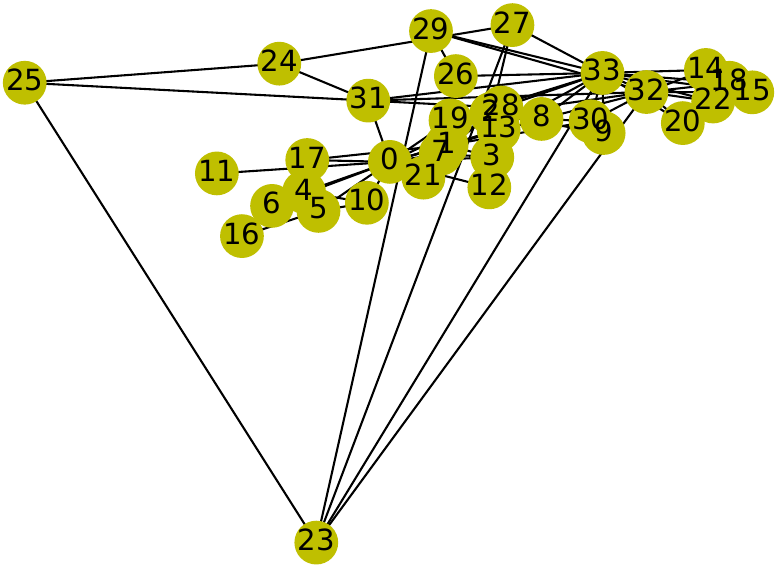}
\end{minipage}
}
\subfigure[]
{
\begin{minipage}[t]{0.45\linewidth}
\centering
\includegraphics[width=4.0cm]{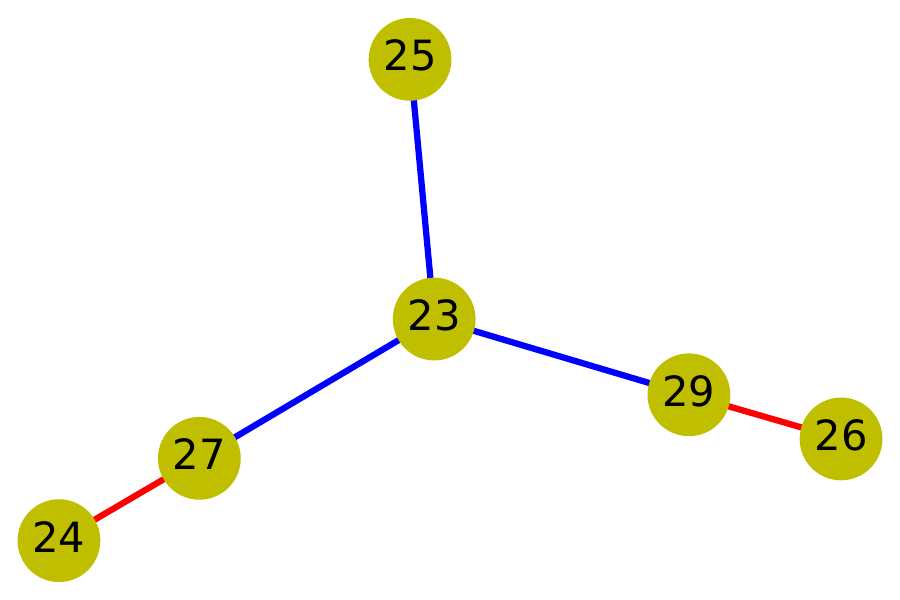}
\end{minipage}
\label{fig:case_study_ggn}
}
\caption{The network structure (a) and the GGN (b) of case study. In the GGN, the color of positive edges is red and the color of negative edges is blue.}
\label{fig:case_study}
\end{figure}

In this section, we will show how to apply GGNs to conduct relationship inference in team games. 

\subsubsection{Experimental Setup}

Based on the fact that most of the time we can only get the information that there is a relationship between two agents, but we do not know which kind of relationship (e.g., positive or negative) between them, we assume that the network $\mathcal{G}$ we get of $n$ agents is unsigned and unweighted, where an edge between $u$ and $v$ means that $u$ and $v$ know each other.
Given the network of the agents and the results of their decision makings, we can build a network in the framework of GGN.

Based on the above considerations, we design the experiments as follows: For a real network $\mathcal{G}'$ which is signed and weighted, we stimulate the team game on $\mathcal{G}'$ and generate teaming results which can be considered as the real strategies the agents take. We use the teaming results and the ideal team game model on $\mathcal{G}$ (without $\mathcal{G}'$) to generate GGN and calculate the $k$th-order exponential kernel of the GGN to infer the relationships. Since $\mathcal{G}'$ is the real network, the weighted adjacency matrix of $\mathcal{G}'$ represents the real relationships between agents, which can serve as the ground-truth to evaluate the performance with respect to relationship inference.

We extract four datasets from Slashdot and Epinions \cite{massa2005controversial,kunegis2009slashdot,leskovec2010signed} (two from each one). Slashdot is a technology news website that lets users tag other users as friends and foes and Epinions is a product review website where users build links that indicate trust or distrust about other users. In order to make the team game converge, we build the network with undirected edges. The statistics of these datasets are listed in Table \ref{table:dataset}.

\begin{figure}[tp!]
\centering
\includegraphics[width=0.79\linewidth]{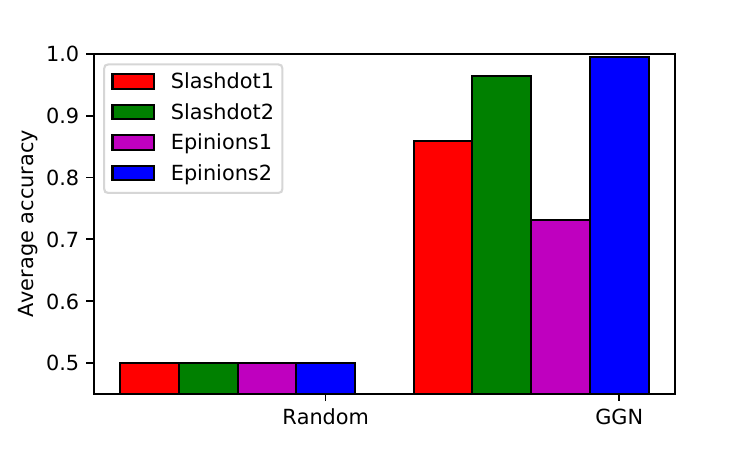}
\caption{Average accuracy of relationship inference in the team game.}
\label{fig:ggn_bar}
\end{figure}

\subsubsection{Case Study}

We use Zachary's Karate Club network \cite{zachary1977information} as a toy example in our experiments. Zachary's Karate Club is a well-known social network of a university karate club with 34 nodes and 78 edges. 
We assume that node 23 has negative relationships with all her neighbors and all other relationships are positive, which leads to the outcome that no one wants to team with node 23. We stimulate the game process and generate the GGN. The structure of the real network and the generated GGN are illustrated in Fig. \ref{fig:case_study}. 

From the snapshot of the GGN in Fig. \ref{fig:case_study_ggn}, we observe that node 25, 27, 29 have first-order negative relationships to node 23, and node 24, 26 have second-order negative relationships to node 23, which is consistent with the real network.

\subsubsection{Relationship Inference for Team Game}

\begin{figure}[tp!]
\centerline{\includegraphics[width=\linewidth]{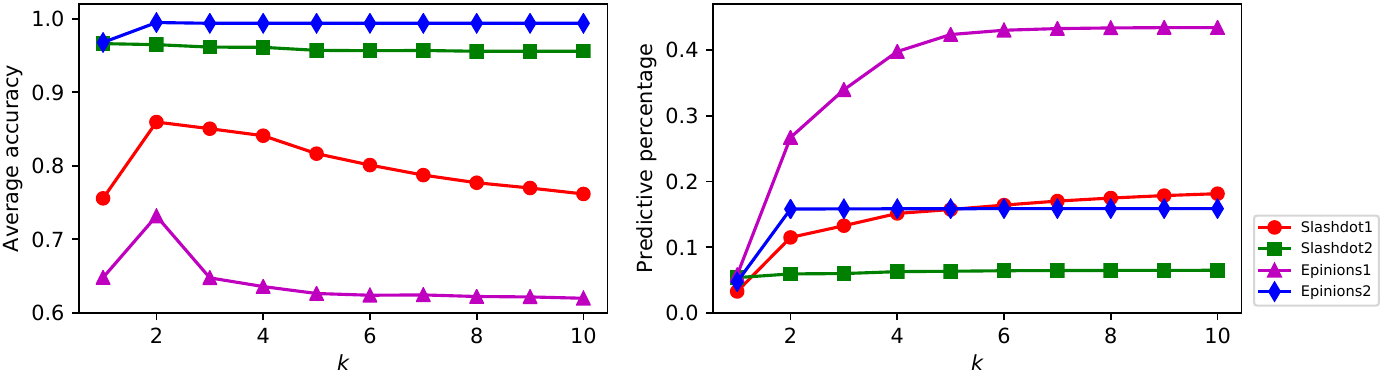}}
\caption{Results of relationship inference in the team game when $c = 0.2$.}
\label{fig:ggn_02}
\end{figure}

\begin{figure}[tp!]
\centerline{\includegraphics[width=\linewidth]{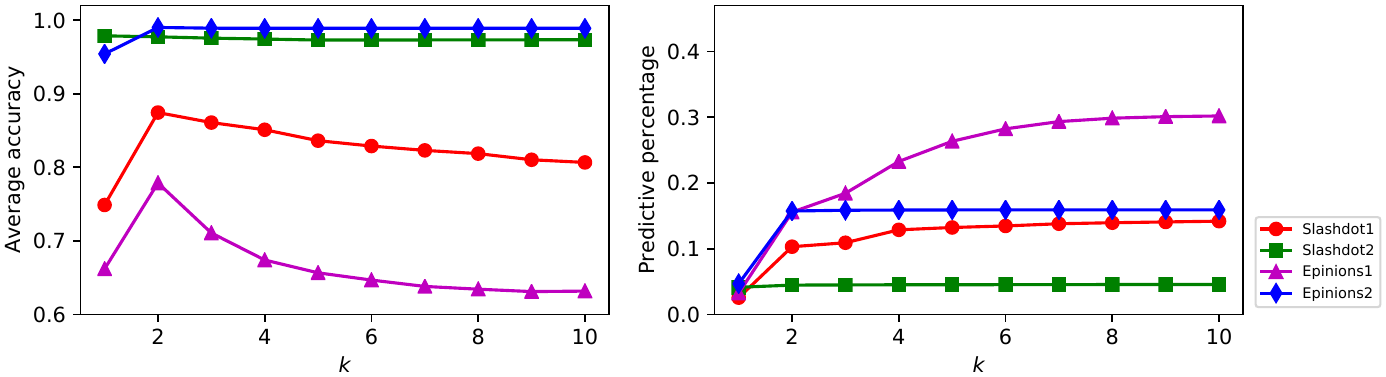}}
\caption{Results of relationship inference in the team game when $c = 0.3$.}
\label{fig:ggn_03}
\end{figure}

In this section, we evaluate the performance of mining on GGN in the relationship inference task. Since the number of positive edges and negative edges is imbalanced in our datasets and for ease of understanding, we use average accuracy (the average accuracy of each class) as the evaluation metric. We set $c = 0.2$ for the team game process and $k = 2$ for the $k$th-order exponential kernel in our network mining approach. 
Here we should mention that our task is brand new and there is no baseline that uses the game outcome, i.e., the strategies taken by agents, as input to infer the relationships between agents, so only random guessing can be compared in our experiments.

From the results shown in Fig. \ref{fig:ggn_bar}, we have several observations. First, the prediction based on the GGN framework far outperforms random guessing, which indicates the feasibility of our model for mining the hidden relationships among agents with game behaviors. Second, comparing with Slashdot1 and Epinions1, Slashdot2 and Epinions2 are sparser, and the average accuracies of Slashdot2 and Epinions2 are higher than those of Slashdot1 and Epinions1 respectively. This is because, in the team game, the complexity of strategies taken by users is related to the complexity of their relationships, a sparser network means that users' strategies are easier to understand, which is compatible with our results.

To observe the results in different settings of team game, we also conduct experiments on team games with different $c$ values by different $k$th-order exponential kernels. The results are shown in Fig. \ref{fig:ggn_02} and Fig. \ref{fig:ggn_03}, where predictive percentage is the proportion of non-zero values in the predicted edges. 
From the results, we find that $k = 2$ is a good trade-off point with relatively high average accuracy and predictive percentage in most cases. Besides, according to the predictive percentages, we find that only a relatively small percentage of relationships could be predicted in a game process. This is because only a small part of relationships is reflected in a single game, but it is still very valuable to infer these relationships to understand this part of agents. In conclusion, the results demonstrate the feasibility of our GGN framework on revealing relationships between agents with game behaviors.

\section{Related Work}

In general, our work is related to game theory \cite{osborne1994course,myerson2013game} and social network analysis \cite{wasserman1994social,scott2017social}. There are a lot of studies on social network discovery, such as link analysis \cite{getoor2005link,liben2007link} and community detection \cite{fortunato2010community}. Although there are some research \cite{chen2010game} applying game-theoretic models to network mining, no one, in turn, applies network mining to game systems. In addition, some work conducts relationship inference tasks on various social networks \cite{diehl2007relationship} or studies the cooperation and competition among agents in game theory \cite{colman2003cooperation,barash2004survival}, but none can infer the hidden relationships between agents with game behaviors.
Owing to the framework of GGN we designed, our work has a strong connection to signed network mining \cite{leskovec2010signed,leskovec2010predicting,kunegis2009slashdot,yang2012friend,tang2016survey}. From the perspective of economics, our framework is related to revealed preference \cite{samuelson1948consumption}, which aims at inferring the preferences of individuals with the observed choices.

Recently, \citeauthor{serrino2019finding} train a multi-agent reinforcement learning agent to find friend and foe in \textit{The Resistance: Avalon}, a hidden role game played by several players. However, their model is just designed for similar hidden role games with several players, and can not be applied to infer the relationships between a large number of social individuals with game behaviors.

\section{Discussion and Conclusion}

In this paper, we propose a novel game generative network framework to build networks for game systems, which is a combination of game theory and network mining. In general, GGNs sever as the intermediaries between game processes and network mining approaches, which provides a new way to model and mine relationships between social agents with game behaviors. 
To show how to apply GGNs to reveal the relationships of agents, we first give an illustration on Prisoners' Dilemma, and then introduce the team game as a concrete example and conduct experiments on it. Indeed, there are many other potential applications of game generative networks, such as mining the hidden features of agents by using advanced technologies, e.g., graph neural networks (GNNs).
For future work, we plan to further improve the game generative network framework and apply it to more real applications.

\bibliographystyle{named}
\bibliography{main}

\end{document}